\documentclass{amsart}
\usepackage{lineno}
\usepackage{standalone}
\usepackage[colorlinks=true]{hyperref}

\usepackage{tikz}
\usetikzlibrary{decorations.pathmorphing,decorations.pathreplacing,arrows,automata,patterns}
\usepackage{hyperref}
\usepackage{amsfonts,amsmath,amssymb,amsfonts}
\usepackage{thm-restate}
\usepackage{xcomment}
\usepackage{xspace}

\usepackage[style=alphabetic,giveninits,maxnames=6,minnames=5,
maxbibnames=99,natbib=true,
backend=bibtex]{ biblatex }
\usepackage{cleveref}
\usepackage{color}

\newcommand{\A}{\mathcal{A}}
\newcommand{\LL}{\mathcal{L}}
\newcommand{\ZZ}{\mathbb{Z}}

\newcommand{\NN}{\mathbb{N}}
\newcommand{\cantor}{{\ensuremath{\{0,1\}^\NN}}}

\newcommand{\orbit}[1]{\ensuremath{\mathcal O\left(#1\right)}}
\newcommand{\pizu}{\ensuremath{\Pi^0_1}\xspace}

\newcommand{\turdegzero}{{\ensuremath{\mathbf{0}}}\xspace}
\newcommand{\turdeg}[1]{{\ensuremath{\mathbf{#1}}}\xspace}
\newcommand{\cone}[1]{{\ensuremath{\mathcal C_{#1}}}\xspace}
\newcommand{\spectrum}[1]{{\ensuremath{\text{Sp}\left( #1 \right)}}\xspace}

\def\presuper#1#2%
{\mathop{}%
  \mathopen{\vphantom{#2}}^{#1}%
  \kern-\scriptspace%
  #2}

\newtheorem{theorem}{Theorem}[section]
\newtheorem{definition}{Definition}[section]

\newtheorem{lemma}[theorem]{Lemma}
\newtheorem{corollary}[theorem]{Corollary}

\begin{document}

\title[Relating word and computational complexity in subshifts]{The relationship between word complexity and computational complexity in subshifts}
\author{Ronnie Pavlov}
\address{Ronnie Pavlov\\
Department of Mathematics\\
University of Denver\\
2390 S. York St.\\
Denver, CO 80208}
\email{rpavlov@du.edu}
\urladdr{www.math.du.edu/$\sim$rpavlov/}
\thanks{The first author gratefully acknowledges the support of NSF grant
  DMS-1500685. The second author gratefully acknowledges the support of ANR grant ANR-12-BS02-0007.}
\renewcommand{\subjclassname}{MSC 2010}
\author{Pascal Vanier}
\email{pascal.vanier@lacl.fr}
\address{Laboratoire d'Algorithmique, Complexit\'e et Logique\\ Universit\'e de
  Paris Est, LACL, UPEC, France}

\begin{abstract}
  We prove several results about the relationship between the word complexity
	function of a subshift and the set of Turing degrees of points of the subshift,
	which we call the Turing spectrum. Among other results, we show that a Turing spectrum 
	can be realized via a subshift of linear complexity if and only if it consists of the 
	union of a finite set and a finite number of cones, that 
	a Turing spectrum can be realized via a subshift of exponential complexity 
	(i.e. positive entropy)	if and only if it contains a cone, and that every
	Turing spectrum which either contains degree $\turdegzero$ or is a union of cones
	is realizable by subshifts with a wide range of `intermediate' complexity growth 
	rates between linear and exponential. 
\end{abstract}
\maketitle

\section{Introduction}\label{S:intro}

In this work, we study various notions of complexity for subshifts. A subshift is a set $X$ of bi-infinite sequences
of symbols from a finite alphabet $\A$ which is invariant under shifts and closed in the product topology;
any such $X$ is then a dynamical system when associated with the shift map
$\sigma$.

Another, equivalent, way to define subshifts is through families of finite forbidden
words: any such set induces the subshift of sequences in which no forbidden word appears. 
Different hypotheses on the set of forbidden words yield classes of subshifts; 
for instance, subshifts of finite type (SFTs) correspond to finite sets of forbidden words, 
and effective subshifts correspond to recursively enumerable sets of forbidden words.


There are many ways to measure how ``complicated'' a
subshift is. One of the simplest is the so-called word
complexity function $(c_n(X))$, where $c_n(X)$ is the number of $n$-letter words appearing within points of $X$.
The minimum unbounded growth rate of the complexity function is linear complexity, and
the maximum growth rate is exponential growth, i.e. positive topological entropy.

A more recent measurement of complexity comes from computability theory,
specifically the notion of Turing degrees.
The Turing degree of a sequence is a measure of its computational power~\cite{Rog1987}, and the
Turing spectrum~\cite{JeandelV2013} is the set of all Turing degrees of points
in a subshift $X$.
In \cite{JeandelV2013}, the authors proved several results about the Turing
spectra of various subshifts, focusing mostly on so-called multidimensional
shifts of finite type. One of the results, however, applied to minimal subshifts in
any dimension: namely, such a subshift's spectrum either consists only of \turdegzero 
 (the degree of computable sequences) or must contain the cone above 
any of its degrees (this result is recalled in~\Cref{lem:recurrentcone}). It was later
proved that spectra of minimal subshifts may contain several cones \cite{HochmanV2016} and
then that they actually correspond exactly to the enumeration cones of co-total
enumeration degrees \cite{McCarthyMinimal}.

It is then natural to wonder about the relationship between these two notions of
complexity, especially since substantial connections between entropy and computability are 
known (see~\cite{HochMey}).

In this work,
we prove several results in this direction. Namely, under various complexity assumptions, we give
necessary conditions on the spectrum and sufficient conditions on a set of
Turing degrees for there to exist a subshift with the desired complexity which realizes that spectrum.
The somewhat surprising summary of these results is that the spectrum is heavily
restricted near minimal (i.e. linear)
and maximal (i.e. exponential) complexity, and these restrictions relax as the
complexity moves further from these extremes.

For subshifts with linear complexity, the following results show that the 
spectrum is comprised of a finite number of cones and
isolated degrees, with size bounded by the `slope' (i.e. linear growth rate). (See \Cref{S:defs} for
formal definitions.)

\begin{restatable*}{theorem}{thmStrongLinearBound}\label{thm:linearcomplexity}
  If a subshift $X$ has strong linear complexity with index $t$, 
	then $\spectrum{X} \diagdown \{\turdegzero\}$ can be written as the union of some number $c$ of
	cones and $m$ other (isolated) degrees, where $c + m < t$.
\end{restatable*}

In the opposite direction, the following theorem shows that any spectrum of this type can be realized by a subshift with linear complexity.

\begin{restatable*}{theorem}{thmRealStrongLinear}
  Let $S$ and $S'$ be  finite sets of Turing degrees. There exists a subshift
  $X$ with strong linear complexity whose spectrum is
  \[\spectrum{X}= \{\turdegzero\} \cup S \cup \bigcup_{\turdeg{d'} \in S'} \cone{\turdeg{d'}}\text{.}\]
\end{restatable*}

The other extreme for complexity is exponential growth, also called positive (topological) entropy, which also influences the
spectrum:
\begin{restatable*}{theorem}{thmPosEntropy}\label{posent}
  If a subshift $X$ has positive entropy, then its spectrum contains a cone. 
\end{restatable*}

Again, a result in the opposite direction holds: every such spectrum can be realized
by a subshift with positive entropy, which can be taken to be arbitrarily large. 
\begin{restatable*}{theorem}{thmRealPosEntropy}
For any subshift $X$ with entropy $h$, if $\spectrum{X}$ contains a cone $C$,
then there exists a subshift $P$ with $\spectrum{P} = \spectrum{X}$ and entropy
$h'$, for any $h'>h$ such that $\deg_T h'\in C$.
\end{restatable*}

Our remaining results treat complexities strictly between linear and exponential.
First, we show that subshifts with complexities arbitrarily close to linear
can realize more spectra than subshifts with linear complexity.

\begin{restatable*}{theorem}{thmRealArbSlow}
  Let $S$ be a countable set of Turing
  degrees, and let $S'$
  be a finite set of Turing degrees. Then the set of 
  subshifts $X$ with $\spectrum{X} = \{\turdegzero\} \cup S \cup \bigcup_{\turdeg{d'} \in S'} \cone{\turdeg{d'}}$ exhibits arbitrarily slow
  superlinear complexity. 

\end{restatable*}

By slightly strengthening the assumption on complexity to require that it
is ``computably close'' to linear, we can realize an extremely large class of
spectra.

\begin{restatable*}{theorem}{thmRealArbSlowComp}\label{thm:compslowsuplin}
  Let $S$ be a closed nonempty set of Turing degrees. Then the sets of 
  subshifts $X$ with $\spectrum{X} = S \cup \{\turdegzero\}$ or
  $\spectrum{X}=\bigcup_{\turdeg{d}\in S}\cone{\turdeg{d}}$ exhibit arbitrarily computably slow
  superlinear complexity.
\end{restatable*}

By a closed set of Turing degrees, we mean here that the set of Turing degrees
is realizable by a closed subset of $\cantor$. Since subshifts are closed,
the Turing spectrum of a subshift must be closed, and so some assumption of this
type on $S$ is necessary.
Closed sets of Turing degrees are of
course a strict subclass of unconstrained sets of Turing degrees: for instance, apart from the
set of all degrees and the empty set, a set and its complement cannot both be
closed sets of Turing degrees (see~\cite{LewisKent}).

Similar results hold for complexity close to exponential.

\begin{restatable*}{theorem}{thmRealArbFastComp}\label{thm:fastcompsubexp}
  Let $S$ be a closed nonempty set of Turing degrees. Then the set of subshifts with
  $\spectrum{X} = S \cup \{\turdegzero\}$ exhibits arbitrarily 
  computably fast subexponential complexity.
\end{restatable*}

\begin{restatable*}{theorem}{thmRealArbFast}\label{thm:fastsubexp}
  Let $S$ be a closed nonempty set of Turing degrees. Then the set of subshifts with
  $\spectrum{X} = \bigcup_{\turdeg{d} \in S} \cone{\turdeg{d}}$ exhibits arbitrarily 
  fast subexponential complexity.
\end{restatable*}

Finally, we show that given an intermediate growth rate, i.e. neither linear nor
exponential, we can can construct subshifts with complexity close to this growth
rate and realising these same classes of spectra.

\begin{restatable*}{theorem}{thmRealIntermediate}\label{thm:intermediate}
  Let $S$ be a closed nonempty set of Turing degrees, and let $g: \mathbb{N}
  \rightarrow \mathbb{N}$ be any computable function
	which is superlinear (i.e. $\frac{g(n)}{n} \rightarrow \infty$) and subexponential (i.e. $\frac{\log g(n)}{n} \rightarrow 0$). 
	Then there exists a subshift $X$ with $\spectrum{X} = \{\turdegzero\} \cup S$ and a sequence $(n_k)$ so that 
	\[
	\lim_{k \rightarrow \infty} \frac{c_{n_k}(X)}{g(n_k)} = 1.
	\]
\end{restatable*}

\begin{restatable*}{theorem}{thmRealIntermediat}\label{thm:intermediat}
  Let $S$ be a closed nonempty set of Turing degrees, and let $g: \mathbb{N}
  \rightarrow \mathbb{N}$ be any computable function
	which is superquadratic (i.e. $\frac{g(n)}{n^2} \rightarrow \infty$) and subexponential (i.e. $\frac{\log g(n)}{n} \rightarrow 0$). 
	Then there exists a subshift $X$ with $\spectrum{X} = \bigcup_{\turdeg{d} \in S} \cone{\turdeg{d}}$ and a sequence $(n_k)$ so that 
	\[
	\lim_{k \rightarrow \infty} \frac{c_{n_k}(X)}{g(n_k)} = 1.
	\]
\end{restatable*}

We note that all of our realization results are for Turing spectra which
either contain the degree $\turdegzero$ or are unions of cones. We also show
that there exist subshifts with spectrum outside of these two categories, 
and so our results do not realize all possible spectra for subshifts. 

\begin{restatable*}{theorem}{thmWeirdSpec}\label{thm:weirdSpec}
 Given any degree $\turdeg{d_1}\leq\turdeg 0'$ and any degree $\turdeg{d_2}$, where
 $\turdeg 0'$ is the degree of the halting problem, there exists a subshift $X$ such that
 $\spectrum{X}=\cone{\turdeg{d_1}}\cup\{\turdeg{d_2}\}$. In particular,
 $\turdeg{d_2}$ is not in $\cone{\turdeg{d_1}}$ if we take
 $\turdeg{d_2}<\turdeg{d_1}$ or $\turdeg{d_2}$ and $\turdeg{d_1}$ incomparable.
 
\end{restatable*}

We give relevant definitions in \Cref{S:defs}, then
restrictions imposed on the Turing spectrum by word complexity in
\Cref{S:restrict}, give realization results in \Cref{S:Real} and end with giving
an example of spectrum that differs from our realizations in \Cref{S:weirdSpec}.

\section{Definitions}\label{S:defs}
\subsection{Subshifts}
Given a finite set $\A$ (called the \textbf{alphabet}), a \textbf{subshift} over $\A$ is a
subset of $\A^{\ZZ}$ which is closed in the product topology and invariant under the shift map $\sigma$
defined by $(\sigma x)(n) = x(n+1)$.

For any $x \in \A^\ZZ$, the \textbf{orbit} of $x$ is the set of all its translates 
$\orbit{x} = \{\sigma^k(x)\mid k\in\ZZ\}$. 

A \textbf{word} is any element of $\A^n$ for some $n \in \NN$. For any word $w$, its \textbf{length}, written $|w|$, is the integer $n$ so that $w \in \A^n$. 
Subshifts can alternately be defined in terms of words; for any set $\mathcal{F}$ of words, one can define the subshift $X(\mathcal{F})$ given by the set of all sequences on $\A$ which do not contain any words of $\mathcal{F}$. A subshift is a \textbf{subshift of finite type (SFT)} if it is $X(\mathcal{F})$ for some finite $\mathcal{F}$.

A point $x \in \A^\ZZ$ is \textbf{periodic} if there exists $p \neq 0$ so that $x(n) = x(n + p)$ for all $n \in \ZZ$, and \textbf{aperiodic} otherwise. We say that $x$ is \textbf{eventually periodic on the left} if there exists $p \neq 0$ and $N$ so that $x(n) = x(n + p)$ for all $n < N$, and \textbf{eventually periodic  the right} if there exists $p \neq 0$ and $N$ so that $x(n) = x(n + p)$ for all $n > N$.

A point $x \in \A^\ZZ$ is \textbf{recurrent} if every word appearing within $x$ appears infinitely many times in $x$. We note that any aperiodic sequence which is eventually periodic on the left and right (e.g. $\presuper{\omega}{(01)} 2 (345)^{\omega}$ is not recurrent, since any portion which contains letters breaking the periodicity on the left and right can appear only finitely many times.

A subshift is \textbf{minimal} if it does not properly contain any nonempty subshift.

For a subshift $X$, the \textbf{language} $\LL(X)$ is the set of all words
appearing within some point in $X$. This notion can be extended to
individual points, i.e. for $x \in X$, $\LL(x)$ is the set of all words appearing in $x$.
The set of all words of $\LL(X)$ of length $n$ will be denoted $\LL_n(X)=\LL(X)\cap \A^n$.

\subsection{Complexity}

The \textbf{complexity function} $c_n(X)$ of a subshift is defined by $c_n(X) = |\LL_n(X)|$, i.e.
$c_n(X)$ is the number of all words with length $n$ appearing in some $x \in X$. 

One of the first results in symbolic dynamics is the Morse-Hedlund theorem, which states that if there exists $n$ for which $c_n(X) \leq n$, then 
$X$ is finite, and in particular all points in $X$ are periodic. Therefore, all nontrivial $X$ satisfy $c_n(X) > n$ for all $n$,
i.e. have complexity growing at least linearly.

A subshift $X$ has \textbf{strong linear complexity (with index $t$)} if there exists $t$ so that $\limsup c_n(X) - tn < \infty$,
and \textbf{weak linear complexity (with index $t$)} if there exists $t$ so that $\liminf c_n(X) - tn < \infty$. We say that
$X$ has (weak/strong) linear complexity if it has (weak/strong) linear complexity for some index $t$.

A collection $S$ of subshifts \textbf{exhibits arbitrarily slow super-linear complexity} if, for every increasing
unbounded $f: \mathbb{N} \rightarrow \mathbb{N}$, there exists $X \in S$ where $c_n(X) < nf(n)$ for sufficiently large $n$.

A collection $S$ of subshifts \textbf{exhibits arbitrarily computably slow super-linear complexity} if, for every computable increasing
unbounded $f: \mathbb{N} \rightarrow \mathbb{N}$, there exists $X \in S$ where $c_n(X) < nf(n)$ for sufficiently large $n$.

A collection $S$ of subshifts \textbf{exhibits arbitrarily fast subexponential complexity} if, for every increasing
unbounded $f: \mathbb{N} \rightarrow \mathbb{N}$, there exists $X \in S$ with $\lim \frac{\log c_n(X)}{n} = 0$, but 
$\frac{\log c_n(X)}{n} > \frac{1}{f(n)}$ for sufficiently large $n$.

A collection $S$ of subshifts \textbf{exhibits arbitrarily computably fast subexponential complexity} if, for every computable increasing
unbounded $f: \mathbb{N} \rightarrow \mathbb{N}$, there exists $X \in S$ with $\lim \frac{\log c_n(X)}{n} = 0$, but 
$\frac{\log c_n(X)}{n} > \frac{1}{f(n)}$ for sufficiently large $n$.

A subshift $X$ has \textbf{positive entropy} if $\lim_{n \rightarrow \infty} \frac{\log c_n(X)}{n} > 0$; this limit is 
called the \textbf{topological entropy} of $X$ and written $h(X)$. For more information on entropy (including a proof
of the existence of the limit), see \cite{walters}.

\subsection{Ergodic theory}\label{subsec:ergodic}

Here we summarize a few results from ergodic theory which will be used in a couple of later proofs; 
see \cite{walters} for a detailed introduction to the subject. Throughout this section,
all measures considered are Borel probability measures on $\A^{\mathbb{Z}}$. Several statements are slightly weaker
than the full strength of the theorem in question; we've stated versions suited to our setting.

\begin{theorem}[Poincar\'{e} recurrence theorem]\label{poincare}
If $X$ is a subshift and $\mu(X) = 1$, then $\mu$-a.e. $x \in X$ is recurrent.
\end{theorem}

\begin{definition}
A measure $\mu$ is \textbf{ergodic} if every measurable $A$ for which $A = \sigma A$ satisfies $\mu(A) \in \{0,1\}$. 
\end{definition}

\begin{theorem}[Birkhoff pointwise ergodic theorem]\label{birkhoff}
If $\mu$ is an ergodic measure and $w$ is a word, then for $\mu$-a.e. $x$,
\[
\lim_{n \rightarrow \infty} \frac{|\{0 \leq i < n \ : \ \sigma^i x \in [w]\}|}{n} \rightarrow \mu([w]),
\]
where $[w]$ is the set of $x$ containing an occurrence of $w$ starting at the $0$th coordinate.
\end{theorem}

\begin{lemma}\label{periodic}
If $X$ is a subshift with positive topological entropy, then there exists an ergodic measure $\mu$ where $\mu(X) = 1$ and
$\mu$-a.e. $x \in X$ is not periodic.
\end{lemma}

\begin{proof}
We prove the contrapositive, and assume that all ergodic $\mu$ with $\mu(X) = 1$ satisfy $\mu(P) > 0$, where $P$ is the set of periodic points of $X$. For every ergodic $\mu$, since $P = \sigma P$ and $\mu(P) > 0$, in fact $\mu(P) = 1$. It is easy to see from the definition that any such $\mu$ has measure-theoretic entropy $h(\mu) = 0$ (the definition of measure-theoretic entropy is long, and so we do not include it here.) By the so-called ergodic decomposition, every measure $\mu$ with $\mu(X) = 1$ can be written as a convex combination of ergodic measures, and since measure-theoretic entropy is linear, $h(\mu) = 0$ for all such $\mu$. Finally, the Variational Principle (\cite{walters}) states that the topological entropy $h(X)$ is the supremum of the measure-theoretic entropy $h(\mu)$ over all $\mu$ with $\mu(X) = 1$, and so $h(X) = 0$, completing the proof.
\end{proof}

\subsection{Computability}
A set $S\subseteq\{0,1\}^\NN$ is called \textbf{effectively closed} or a \textbf{\pizu
class} iff there exists a Turing
machine $M$ such that the sequences of $S$ are exactly the ones on which $M$
does not halt. Equivalently, there is a Turing machine that enumerates the
prefixes of all sequences not in $S$.

An effectively closed subshift, or \textbf{effective subshift} is thus a subshift 
$X(\mathcal{F})$ for some recursively enumerable set $\mathcal{F}$ of forbidden words. 
SFTs are effectively closed, but not all effectively closed subshifts are SFTs.

Define the partial order $\leq_T$ on $\{0,1\}^\NN$ as follows: for $x,y\in\{0,1\}^\NN$, $x\leq_T y$ 
if there exists a Turing machine which, given $y$ as an oracle, outputs $x$. We say that $x\equiv_T y$
whenever $x\leq_T y$ and $y\leq_T x$, and since $\leq_T$ is clearly reflexive and transitive, 
$\equiv_T$ is an equivalence relation.
A \textbf{Turing degree} is an equivalence class for this relation. Given a
sequence $x\in\{0,1\}^\NN$, we denote its Turing degree by $\deg_T x$. The computable
sequences all share the lowest degree, denoted by \turdegzero.
For any degree $\turdeg{d}$, the \textbf{cone with \textbf{base} \turdeg{d}} is the 
set $\left\{ \turdeg e \mid \turdeg d\leq_T \turdeg e \right\}$ of all degrees above $\turdeg{d}$, 
and is denoted by $\cone{d}$.
We will often write $\cone{x}$ instead of the more formally correct $\cone{\deg_T x}$.

The \textbf{Turing spectrum} of a subshift $X$, denoted by \spectrum{X}, is the set 
$\{\turdeg d\mid \exists x\in X,\deg_T(x)=\turdeg d\}$
of all Turing degrees of all points of $X$.

The following lemma describes a simple hypothesis guaranteeing that $\spectrum{X}$ contains a cone. 
It is essentially the same as Theorem 5.3 from \cite{JeandelV2013}, and we present 
an abbreviated proof describing why the hypothesis of minimality of $X$ can be replaced by aperiodicity 
and recurrence of $x$.

\begin{lemma}[\cite{JeandelV2013}]\label{lem:recurrentcone} 
  For any subshift $X$ and any aperiodic recurrent point $x\in X$, 
	$\cone{x}\subseteq \spectrum{X}$.
\end{lemma}
\begin{proof}
  \newcommand{\enc}{\ensuremath{\mathbf{enc}}\xspace}
  \newcommand{\dec}{\ensuremath{\mathbf{dec}}\xspace}
  
Since $x$ is aperiodic and recurrent, it cannot be both eventually periodic on the left and right. Since 
  reflecting a sequence about the origin does not affect its Turing degree, we may assume without loss
  of generality that $x$ is not eventually periodic on the right.
    
  We construct two functions $\enc: \cantor \to X$ and
  $\dec: \enc(\cantor) \to \cantor$ such that $\enc(y)$ is an ``encoding'' of $y$
  derived from $x$ and $\dec$ is a ``decoding,'' with the property that $\dec(\enc(y)) = y$ for all $y$.

  Let $\prec$ be a fixed order on $\Sigma$ the alphabet of $X$. We compute $\enc(y)$
  in an inductive way:
  \begin{itemize}
  \item We start with $w_{0}=x_0$ the letter in the center of $x$.
  \item For each $k \geq 0$, the word $w_k$ will be a subword of $x$. 
	Then, by Lemma 5.2 from \cite{JeandelV2013},
	there exist subwords $w_{k,0}$ and $w_{k,1}$ of $x$ with the following properties:
	$w_k$ occurs exactly twice within each $w_{k,i}$, the first differing letters $a,b$ following the
	left/right occurrences of $w_k$ in $w_{k,0}$ satisfy $a \prec b$, and the corresponding letters
	$e,f$ for $w_{k,1}$ satisfy $f \prec e$. (The proof of Lemma 5.2 was stated for minimal $X$,
	but clearly only used the facts that $x$ is recurrent and not eventually periodic on the right.)

	
    Depending on whether $y_k = 0$ or $1$, consider the 
		word $w_{k, y_k}$, and take the occurrence of this word within $x$ closest to the origin.
    Let $i\in \ZZ$ be the position within $x$ of the first letter of the first $w$ within this
		$w_{k, y_k}$, and let $j$ be the position within $x$ of $b$ or $f$, i.e. the rightmost of the
		first differing letters following $w_k$ within $w_{k, y_k}$. Finally,
    define $w_{k+1} = x_{[i-(j-i);j]}$.
  \end{itemize}
  As $w_k$ is a subword in the center of $w_{k+1}$, the sequence
  $(w_k)_{k\in\NN}$ converges, and the sequence to which it converges is
  $\enc(y) \in X$.

  The decoding $\dec$ is now straightforward: for any $z \in \enc(\cantor)$, 
  we start by looking at the letter in the center, which recovers $w_0$ from
	the definition of $\enc$. We then look for its next 
  occurrence to the right, and we then 
  scan the letters following both until we find a difference (this process succeeds
	by definition of $\enc$.) Once one 
  is found, we check which case it is and recover $y_0$. This allowed us to recover $w_1$, and so
  we now do the same procedure again starting with $w_1$ and so on.
  
  Clearly $\dec$ is computable, and $\enc$ is computable given $x$. Therefore, for
	any $y \in \cantor$, $y\leq_T \enc(y)\leq_T \sup(x,y)$. This means that if $y \geq_T x$, 
	then $\deg_T\enc(y)=\deg_T y$. Since $\enc(y) \in X$ and $y \geq_T x$ was 
  arbitrary, this shows that the spectrum of $X$ contains the cone above $\deg_T x$.
\end{proof}

The following corollary is nearly immediate.

\begin{corollary}\label{lem:coneor0} 
  For any subshift $X$, $\spectrum{X}$ contains either a cone or $0$.
\end{corollary}

\begin{proof}
  By \Cref{poincare}, $X$ contains a recurrent point $x$. If $x$ is periodic, then obviously
  $\deg_T(x) = \turdegzero \in \spectrum{X}$. If $x$ is aperiodic, then by Lemma~\ref{lem:recurrentcone}, 
  $\cone{x} \subseteq \spectrum{X}$.
\end{proof}

\subsection{Sturmian subshifts}\label{S:Sturmians}

A particular class of subshifts which will be relevant for several of our results are the so-called Sturmian subshifts. 
For any $\alpha \in (0,1) \cap \mathbb{Q}^c$, the \textbf{Sturmian subshift with rotation number $\alpha$} is defined as follows:
\[
  S_{\alpha} := \overline{\{ \left(\lfloor (n+1)\alpha + c \rfloor - \lfloor n\alpha + c \rfloor \right)_{n \in \mathbb{Z}} \ : \ 
    c \in \mathbb{R}\}}.
\]
The relevant properties of Sturmian subshifts for our purposes are the following.
\begin{itemize} 
\item $c_n(S_{\alpha}) = n+1$ for all $n \in \mathbb{N}$. 
\item Every $S_{\alpha}$ is minimal.
\item For every $n$ and every $w \in \LL_n(S_{\alpha})$, the number of $1$s in $w$ is either $\lfloor n\alpha \rfloor$ or $\lfloor n\alpha \rfloor + 1$.
\end{itemize}

The third property above implies that $\spectrum{S_{\alpha}} \subset \cone{\deg_T \alpha}$; for any $x \in S_\alpha$,
the frequency of $1$s in the first $n$ letters of $x$ is within $\frac{1}{n}$ of $\alpha$, and 
so $\deg_T(x) \geq_T \deg_T(\alpha)$. Conversely, the point in $S_\alpha$ corresponding to $c = 0$ in the definition
above clearly has degree $\deg_T \alpha$, and so by Lemma~\ref{lem:recurrentcone},
$\spectrum{S_{\alpha}} \supset \cone{\deg_T \alpha}$. We conclude that $\spectrum{S_{\alpha}} = \cone{\deg_T \alpha}$ for any Sturmian subshift $S_\alpha$.

\section{Restrictions imposed by complexity on the (Turing) spectrum}\label{S:restrict}

In this section, we prove several results about restrictions on $\spectrum{X}$ given knowledge about the growth rate of the word complexity function $c_n(X)$. Somewhat unsurprisingly, the largest restriction is imposed when $X$ has linear complexity, the slowest possible (nontrivial) growth rate.

\subsection{Implications of linear complexity}\label{subsec:linearbound}

When $c_n(X)$ grows as slowly as possible, namely like $n$ plus a constant, $\spectrum{X}$ is extremely restricted.

\begin{theorem}
  If $X$ is a subshift and $\limsup c_n(X) - n < \infty$, then \spectrum{X} is
  equal to either the singleton $\turdegzero$, a single cone, or the union of
  $\{\turdegzero\}$ and a single cone.
\end{theorem}
\begin{proof}

By \cite{coven}, if $\limsup c_n(x) - n < \infty$, then $x$ is either periodic, eventually periodic in both directions but aperiodic (e.g. $\ldots 00012323 \ldots$), or quasi-Sturmian, i.e. the image of a Sturmian sequence under a shift-commuting homeomorphism. In this third case, $x$ is aperiodic and recurrent (and in fact $X$ is minimal). 

Consider any $X$ as in the theorem. We claim that $X$ cannot contain two subsystems of the third (i.e. quasi-Sturmian) type. Indeed, unequal minimal systems are disjoint, and so if $X$ contained two quasi-Sturmian systems $Y, Z$, they would be disjoint. Then for large enough $n$, $\LL_n(Y)$ and $\LL_n(Z)$ would be disjoint, implying $c_n(X) \geq 2n$ for large enough $n$, contradicting the original assumption on $X$. Clearly, all $x$ of the first two types (periodic or eventually periodic in both directions) yield only points in $X$ with degree $\turdeg{0}$. It therefore remains only to show that a quasi-Sturmian system has spectrum equal to a single cone. This follows immediately from the fact that the spectrum of a Sturmian system is a single cone, and the fact that a shift-commuting homeomorphism is a computable map.
   

\end{proof}

For subshifts of linear complexity, the spectrum is still heavily constrained by the growth rate of $c_n(X)$.

\thmStrongLinearBound
The proof of \Cref{thm:linearcomplexity} will be divided into two parts. In
\Cref{lem:linearfinitelonedegrees}, we prove that $\spectrum{X}$ 
contains some number $m < t$ of nonzero isolated degrees whose associated cones
are not subsets of $\spectrum{X}$.

Then, 
we prove that $\spectrum{X}$ is the union of these isolated 
degrees and some number $c < t - m$ of cones. Interestingly, most portions of this proof
require only the weaker assumption of weak linear complexity.

\subsubsection{Linear complexity and number of isolated degrees}

From \Cref{lem:recurrentcone}, we know that isolated degrees in the spectrum of
a subshift, degrees that are not part of a cone, can only come from
non-recurrent points.

\begin{lemma}\label{lem:linearfinitelonedegrees}
 If a subshift $X$ has weak linear complexity with index $t$, 
	then there are fewer than $t$ different orbits of non-recurrent points in $X$.
\end{lemma}
 As a straightforward corollary, we can bound the number of degrees
 that are not the base of some cone:
\begin{corollary}\label{cor:linearfinitelonedegrees}
  If a subshift $X$ has weak linear complexity with index $t$, then there are fewer than $t$ nonzero degrees in $\spectrum{X}$ 
  whose associated cones are not subsets of $\spectrum{X}$.
\end{corollary}


\begin{proof}[Proof of \Cref{lem:linearfinitelonedegrees}]
  Suppose for a contradiction that there are at least
  $t$ non-recurrent points $x_1, \ldots, x_t \in X$ with distinct orbits.
  We may then define words $w_1,
  \ldots, w_t$ s.t. $w_i$ appears in $x_i$ exactly once for each $i$. Note that
  we can freely extend any $w_i$ to a larger subword of $x_i$ without losing
  this property.

  We claim that by extending the $w_i$ if necessary, we may impose the following additional property: for all $j \neq k$, either $w_j$ is not a subword of $x_k$ or $w_k$ is not a subword of $x_j$. Again, we note that once this property holds for a pair $(j,k)$, it is not lost if any $w_i$ is extended to a larger subword of $x_i$. It then suffices to show that for each fixed pair $(j,k)$ in turn, we may extend 
  $w_j$ and/or $w_k$ to larger words for which the desired property holds. 

  Fix any $j \neq k$. If $w_k$ is not a subword of $x_j$, then the desired property already holds. Similarly, if $w_k$ appears at least twice as a subword of $x_j$, then $x_j$ contains a subword containing $w_j$ and two occurrences of $w_k$, which itself cannot be a subword of $x_k$ since $w_k$ appears only once in $x_k$. Then, the desired property could be obtained by extending $w_j$ to this larger word. 
  The only remaining case is that $w_k$ appears exactly once in $x_j$. We claim that $w_k$ can then be extended to a larger subword of $x_k$ which does not appear in $x_j$ at all. Suppose for a contradiction that this is not the case, i.e. every subword of $x_k$ containing $w_k$ also appears in $x_j$. Since $w_k$ appears exactly once in each of $x_j$ and $x_k$, this would mean that for every $n$, the $n$ letters preceding and following $w_k$ are the same in $x_j$ and $x_k$, and so that $x_j$ is just a shift of $x_k$. We know that this is false, and so have a contradiction, implying that $w_k$ can be extended to a longer subword of $x_k$ for which the desired property holds. 

  By earlier observations, since $j \neq k$ were arbitrary, we may now assume without loss of generality that $w_i$ appears exactly once in $x_i$ for each $i$, and that for all $j \neq k$, either $w_j$ is not a subword of $x_k$ or $w_k$ is not a subword of $x_j$. By shifting the $x_i$ and/or extending the $w_i$ if necessary, we may also assume that there exists $N$ so that $x_i(-N) \ldots x_i(N) = w_i$ for each $i$. Now, for any $n > 2N$, define the following words in $\LL_n(X)$: $v_{i,m} = x_i(m) \ldots x_i(m + n - 1)$ for $1 \leq i \leq t$ and $N - n < m \leq -N$. Note that by definition, $v_{i,m}$ is always a subword of $x_i$ containing $w_i$. We claim that all $v_{i,m}$ are distinct. First, note that for $j \neq k$ and any $m, m'$, $v_{j,m}$ contains $w_j$ and is a subword of $x_j$, and $v_{k,m'}$ contains $w_k$ and is a subword of $x_k$. If $v_{j,m} = v_{k,m'}$, then $w_j$ is a subword of $x_k$ and $w_k$ is a subword of $x_j$, which we know not to be the case. Therefore, $v_{j,m} \neq v_{k,m'}$ whenever $j \neq k$. It remains only to show that $v_{i,m} \neq v_{i,m'}$ whenever $m \neq m'$. To see this, note that each of $v_{i,m}$ and $v_{i,m'}$ contains $w_i$ exactly once, and that $w_i$ appears in each at different locations. This completes the proof that all $v_{i,m}$ are distinct, which implies that 
  $c_n(X) \geq t(n - 2N) = tn - 2tN$ for all $n > 2N$. Therefore, $\liminf
  c_n(X) - tn > -\infty$, contradicting weak linear complexity of $X$ with index $t$.
  The assumption that ${X}$ contains at least $t$ orbits of non-recurrent points was then false, completing the proof.
\end{proof}

\subsubsection{Linear complexity and number of cones}
Now that we have control over the degrees of non-recurrent points, we focus
on showing that the remaining nonzero degrees can be written as a finite union of cones.


In order to prove this, we will need several intermediate steps. The first
is to show that there are fewer than $t$ possibilities for the set of subwords of
an aperiodic point of $X$. For $x \in X$, we denote its set of subwords by $\LL(x)$, and
accordingly define $c_n(x)$ and the notions of strong/weak linear complexity of $x$.

\begin{lemma}\label{lem:aperiodicfinitelangs}
  If a subshift $X$ has weak linear complexity with index $t$, and $m$ is the finite number of orbits of non-recurrent points 
  in $X$ guaranteed by \Cref{lem:linearfinitelonedegrees},
	then the set
  \[\left\{ \LL(x)\mid x\in X
      \text{ is aperiodic and recurrent} \right\}\]
  has cardinality less than $t - m$.
\end{lemma}
\begin{proof}
 Recall from the proof of \Cref{lem:linearfinitelonedegrees} that there exists $N$
 so that for all $n > 2N$, there are at least $m(n - 2N)$ words in $\LL_n(X)$ 
 which each contain some word in $\LL_{2N+1}(X)$ only once (these were the
 words $w_i$ in that proof.)

 Now, suppose that $x_1,\dots,x_k$ have different languages and are recurrent and aperiodic. There
 exists an $N' \geq 2N + 1$ such that $\LL_{N'}(x_1),\dots,\LL_{N'}(x_k)$ are all distinct. By recurrence
 of the $x_i$, there exists a window $M$ such that every word in $\LL_{N'}(x_i)$ appears at least twice
 in each $v_i := x_i([-M,M])$.

 We now claim that for every $n>N'$, there are at least $n-2M$ words in $\LL_n(x_i)$ that contain 
 $v_i$. To see this, choose any $1 \leq i \leq k$.
 \begin{itemize}
 \item Suppose that every $n$-letter subword of $x_i$ contains $v_i$. By the Morse-Hedlund theorem, since $x_i$ 
 is aperiodic, $c_n(x_i) > n$, immediately yielding the desired result. 
 \item Otherwise, there exists an $n$-letter subword of $x_i$ not containing $v_i$.
   Without loss of generality (by shifting if necessary), assume that 
	 $x_i(1) \ldots x_i(n)$ does not contain $v_i$ and is immediately preceded or followed
	 by $v_i$; we treat only the former case, as the latter is trivially similar.
	 Then consider the $n$-letter subwords of $x_i$ defined by $u_j = x_i(-j) \ldots x_i(n - j - 1)$ 
	 for $2M \leq j < n$. Each $u_j$ contains its rightmost occurrence of $v_i$ starting
	 at its $(j - 2M + 1)$th letter, and so all $u_j$ are distinct, verifying the claim.
 \end{itemize}
 For each of the $x_i$, denote by $L_i$ this collection of words. 
 The $L_i$ are all disjoint; for $x \in L_i$, $\LL_{N'}(x) = \LL_{N'}(x_i)$,
 since $x$ contains $v_i$. Therefore, $\LL_n(X)$ contains at least $k(n - 2M)$ words
 which contain some $v_i$ and therefore contain every word in $\LL_{N'}(X)$ at least twice. 
 Since $N' \geq 2N+1$, by the above, there are at least $m(n - 2N)$ words in
 $\LL_n(X)$ which contain some word in $\LL_{N'}(X)$ only once. 
 Therefore, for $n>4M$, $c_n(X) > k(n-2M) + m(n - 2N)$, and so by weak linear complexity of $X$
 with index $t$, $k < t - m$.
\end{proof}

We will now show that the set of nonzero degrees of recurrent points of $X$ 
is precisely the union of the cones above the degrees of the finite set $\{\LL(x) \ : \ x \in X\}$.

\begin{lemma}\label{lem:linearrecurrentabovelanguage}
  If $X$ is a subshift with strong linear complexity, then for every $x\in X$, $\deg_T x\geq_T \deg_T \LL(x)$.
\end{lemma}

In order to prove this result, we will need the notion of a \textbf{right special word}
within a sequence $x$: a word $w \in \LL(x)$ is \textbf{right special for $x$} if there exist at least two ways
to extend it on the right, i.e. letters $a \neq b$ such that $wa\in\LL(x)$
and $wb\in \LL(x)$. The degree $d_x(w)$ of $w$ a right special word for $x$ is the number of
additional such extensions it has: a right special word with $d+1$ extensions in $\LL(x)$ has
degree\footnote{It should always be clear from context whether the degree in
  consideration is a degree of a right special word or a Turing degree since
  Turing degrees of finite words are always \turdegzero.} $d$. It follows that for any 
	$n$, the sum of the degrees of all right special words of length $n$ is $c_{n+1}(x) - c_n(x)$.
	It's obvious then that if $c_{n+1}(x) - c_n(x)$ is bounded, then $x$ has strong linear complexity.
	Surprisingly, the converse is also true, as proved by Cassaigne.
	

\begin{theorem}[Cassaigne~\cite{Cassaigne95}]\label{thm:cassaigne}
  If $x$ is a sequence with strong linear complexity, then 
	$c_{n+1}(x) - c_n(x)$ is bounded from above by a constant.
\end{theorem}

In fact, the application of \Cref{thm:cassaigne} is the only place where we require the assumption of
strong (rather than weak) linear complexity.

\begin{proof}[Proof of \Cref{lem:linearrecurrentabovelanguage}]

  We now want to prove that from any point $x\in X$ we can compute $\LL(x)$.
  By \Cref{thm:cassaigne}, there exist
  $N,M$ such that for every $n>N$, $c_{n+1}(x) - c_n(x) \leq M$, and for infinitely many $n$,
  $c_{n+1}(x) - c_n(x) = M$. This means that for a given $n>N$, once we find a set $S$ of right
  special words of length $n$ where $\sum_{w\in S} d_x(w)=M$, then we know we have found
	all possible right special words for $x$ of length $n$, and all possible letters which can follow them in $x$.

  Our algorithm is the following: for input $n$, the algorithm scans larger
	and larger portions of $x$, looking for words of the form $wa$ and $wb$ for
	$w$ of length at least $n$. It keeps a list of such words $w$ it has found, along with
	the letters that can follow $w$, and whenever it begins a new scan of a longer portion
	of $x$, it increments the maximum length of $w$ in its search (beginning by only looking
  for $w$ of length exactly $n$, then $n$ or $n+1$, and so on.) 
	
	By definition of $M$, it will eventually find $k>n$ and a set $S$ of right-special words for $x$ of 
	length $k$ with extensions yielding degrees which sum to $M$, meaning that the algorithm 
	has the complete set of right special words for $x$ of length $k$ and their followers.

  We now claim that a further algorithm can recover the entire $k$-language $\LL_k(x)$. 
	In the sequence $x$, every word that is not right special has exactly one possible follower. 
	This means that every $k$-letter word appearing in $x$ which is not right special forces the following letter, 
	which in turn forces another letter, until one arrives at a right special word in $S$.
	Therefore, every word in $\LL_k(x)$ is a subword of some word of the form $saw$, where $s \in S$, $a \in \A$,
	$w$ is forced by $s$ and $a$, and $saw$ ends with a word in $S$.  
	The algorithm can then just search $x$ for occurrences
	of every possible $sa$ for $s \in S$ and $a \in \A$ and follow them to the right until hitting another word in $S$;
	$\LL_k(x)$ is then the set of all $k$-letter words encountered in this process.
  Finally, $\LL_n(x)$ is just the set of subwords of $\LL_k(x)$, and so the algorithm 
	has computed $\LL_n(x)$ from $x$; since $n$ was arbitrary, the proof is complete.

\end{proof}

\begin{lemma}\label{lem:linearrecurrentbelowlanguage}
 For any subshift $X$ and aperiodic recurrent $x\in X$, 
 there exists an aperiodic recurrent $y\in X$ such that $\LL(x)=\LL(y)$ and $\deg_T y \leq_T \deg_T \LL(x)$. 
\end{lemma}
\begin{proof}
  We give a simple algorithm to construct $y$ from $\LL(x)$. 
	Define $w_1$ to be the minimal word (in the lexicographic order) in $\LL(x)$ which
	contains all words of $\LL_1(x)$ at least twice (such words exist by recurrence of $x$). 

  For all $i > 0$, define $w_{i+1}$ to be the minimal word in $\LL(x)$ which
	contains all words of $\LL_{|w_i|}(x)$ at least twice 
	and which contains $w_i$ in its center (again using recurrence of $x$). 
	The sequence $(w_i)$ then converges to a point, which we call $y$. 
	Since $y$ contains all words in $\LL_{|w_i|}(x)$ for all $i$, and $y$
	is a limit of words in $\LL(x)$, $\LL(y) = \LL(x)$. Aperiodicity of $y$
	follows immediately from this fact and aperiodicity of $x$. Every word
	in $\LL(x)$ occurs at least twice in $y$ by construction, and so $y$ 
	is recurrent. Finally, $y$ was computed from $\LL(x)$, and by 
	\Cref{lem:linearrecurrentabovelanguage}, $\deg_T y \geq_T \deg_T \LL(y) = \deg_T \LL(x)$.
	Therefore, $\deg_T y = \deg_T \LL(x)$.
\end{proof}

The following is an immediate corollary of \Cref{lem:linearrecurrentabovelanguage} and \Cref{lem:linearrecurrentbelowlanguage}.

\begin{corollary}\label{cor:linearrecurrentequallanguage}
 If $X$ is a subshift with strong linear complexity, then for every aperiodic recurrent $x\in X$, 
 there exists an aperiodic recurrent $y\in X$ such that $\LL(x)=\LL(y)$ and $\deg_T y = \deg_T \LL(x)$. 
\end{corollary}

We may now prove \Cref{thm:linearcomplexity}.

\begin{proof}[Proof of \Cref{thm:linearcomplexity}]
  By \Cref{lem:linearfinitelonedegrees}, there exist only finitely many nonzero degrees 
	$\turdeg{d_1}, \ldots, \turdeg{d_m} \in \spectrum{X}$ for $m < t$ associated to non-recurrent points in $X$.
	Then by \Cref{lem:recurrentcone}, for every $\turdeg{d} \in \spectrum{X} \diagdown \{\turdegzero\}$ 
	not equal to any $d_i$, $\cone{d} \subseteq \spectrum{X}$. 
	
  By \Cref{lem:aperiodicfinitelangs}, there exist $x_1, \ldots, x_c$ for $c < t - m$
	so that for all aperiodic non-recurrent $x \in X$, $\LL(x) = \LL(x_i)$ for some $i$. For every 
	$1 \leq i \leq c$, by \Cref{cor:linearrecurrentequallanguage}, there exists
	aperiodic recurrent $y_i$ with $\LL(y_i) = \LL(x_i)$ and $\deg_T(y_i) = \deg_T(\LL(x_i))$. 
	
	Then, by \Cref{lem:recurrentcone}, $\cone{y_i} \subset \spectrum{X}$ for $1 \leq i \leq k$,
	and so $\{\turdeg{d_1}, \ldots, \turdeg{d_m}\} \cup \bigcup_{i=1}^c \cone{y_i} \subset \spectrum{X}$.
  Conversely, for every aperiodic recurrent $x \in X$, by \Cref{lem:linearrecurrentabovelanguage}, 
	$\deg_T x \geq_T \deg_T \LL(x)$. There exists $i$ so that $\LL(x) = \LL(y_i)$, and so in fact
	$\deg_T x \geq_T \deg_T \LL(y_i) = \deg_T y_i$. Therefore, the set
	of degrees of aperiodic recurrent $x \in X$ is contained in $\bigcup_{i=1}^c \cone{y_i}$, meaning that
	$\spectrum{X} \diagdown \{\turdegzero\} \subset \{\turdeg{d_1}, \ldots, \turdeg{d_m}\} \cup \bigcup_{i=1}^c \cone{y_i} \subset \spectrum{X}$, 
	and so that $\spectrum{X} \diagdown \{\turdegzero\} = \{\turdeg{d_1}, \ldots, \turdeg{d_m}\} \cup \bigcup_{i=1}^c \cone{y_i} \subset \spectrum{X}$,
	completing the proof.
	
\end{proof}
\subsubsection{Minimal subshifts with linear complexity have spectrum equal to a cone}
The preceding results also have strong consequences for minimal subshifts with linear complexity;
the spectrum of any such shift is a single cone, whose base is the degree of its language.

\begin{theorem}\label{thm:minimallinearonecone}
  If an infinite minimal subshift $X$ has strong linear complexity, then $\spectrum{X}=\cone{\deg_T \LL(X)}$.
\end{theorem}


\begin{proof}
	Consider an infinite minimal subshift $X$ with strong linear complexity. For every $x \in X$, $\LL(x) = \LL(X)$ (otherwise, 
	the orbit closure of $x$ would be a nonempty proper subsystem of $X$, contradicting minimality), 
	and $x$ is aperiodic and recurrent (otherwise, the orbit closure of $x$ would contain a point not containing all
	words in $\LL(X)$).
	Therefore, by \Cref{cor:linearrecurrentequallanguage}, there exists a point $y\in X$
  with $\deg_T y = \deg_T\LL(X)$.
	Also, for all $x \in X$, \Cref{lem:linearrecurrentabovelanguage} implies that
	$\deg_T x \geq_T \deg_T \LL(X) = \deg_T \LL(y)$, and so $\spectrum{X} \subset \cone{\deg_T \LL(X)}$.  
	Finally, by \Cref{lem:recurrentcone}, $\spectrum{X} \supset \cone{\deg_T \LL(X)}$, completing the proof.
\end{proof}

\subsection{Implications of exponential complexity}

In the opposite direction, if $X$ has exponentially growing complexity function (i.e. $X$ has positive topological entropy), this also restricts $\spectrum{X}$.

\thmPosEntropy
\begin{proof}
  Suppose that $h^{top}(X) > 0$. Then by \Cref{periodic}, there exists a measure $\mu$ giving zero measure to the set of periodic points of $X$. By \Cref{poincare}, $\mu$-a.e. $x \in X$ is recurrent, and so there exists an aperiodic recurrent point $x \in X$. Then, by Lemma~\ref{lem:recurrentcone}, $\spectrum{X} \supset \cone{x}$. 
\end{proof}


\section{Realizing (Turing) spectrums for various growth rates of complexity functions}\label{S:Real}

In this section, we will present examples of subshifts with various complexity restrictions which exhibit different sets of Turing degrees as spectra.
In several cases, our results come close to showing that the restrictions from Section~\ref{S:restrict} are tight (i.e. cannot be relaxed). The main obstacle 
is that every spectrum we can realize is either a union of cones or contains $\turdegzero$.  

\subsection{Examples with strong linear complexity}

\begin{theorem}\label{finitecones}
  Let $S = \{\turdeg{d_1}, \ldots, \turdeg{d_t}\}$ be a finite set of Turing degrees. Then there exists 
  a subshift $X$ having strong linear complexity whose spectrum is
  \[
    \bigcup_{i = 1}^t \cone{\turdeg{d_i}}\text{.}
  \]

\end{theorem}
\begin{proof}
  Given a real $\alpha \in \cantor$, the Sturmian subshift $S_\alpha$
  corresponding to the rotation of the circle by $\alpha$ has complexity function
  $c_n(S_{\alpha})=n+1$ and its spectrum is $\cone{\alpha}$.

  Now take $\alpha_1,\dots,\alpha_t$ sequences of degrees
  $\turdeg{d_1},\dots,\turdeg{d_t}$. Then $X=\bigcup_{i=1}^t X_{\alpha_i}$ is a
  subshift and has strong linear complexity (with index $t$), and its spectrum is exactly $\bigcup_{i=0}^{t} \cone{\turdeg{d_i}}$.
\end{proof}

\begin{theorem}\label{finiteset}
  Let $S=\{\turdeg{d_1}, \dots, \turdeg{d_t}\}$ be a finite set of Turing
  degrees. Then there exists a subshift $X$ with strong linear complexity whose spectrum is
  \[\spectrum{X}= S\cup\{\turdegzero\}\text{.}\]
  Furthermore, if $S$ can be realized by a \pizu class, then $X$ is effective.
\end{theorem}
\begin{proof}
  Take $s_1,\dots,s_t$ sequences of $\{1,2\}^\NN$ with degrees $\turdeg{d_1},\dots,\turdeg{d_t}$
  respectively. Choose a strictly increasing computable function $(m_i)_{i\in\NN}$.
  We define a subshift $X\subset\{0,1,2\}^\ZZ$ consisting of the closure of the union of 
  the orbits of the following sequences $\{x_i\}_{i = 1}^t$:
  \[
    \begin{array}{c}
      x_{i} := \quad \presuper{\omega}{0} . s_{i}(1) 0^{m_1} s_{i}(2) 0^{m_2}s_{i}(3) 0^{m_3}s_{i}(4) 0^{m_4}s_{i}(5) 0^{m_5}s_{i}(6) 0^{m_6}\cdots 
    \end{array}
  \] 
  Each $x_i$ is just $s_i$ with growing numbers of $0$ interspersed between its letters, and so since $(m_i)$ is computable,
  $\deg_T(x_i) = \deg_T(s_i)$ for $1 \leq i \leq t$. The only points in $X$ other than the orbits of the $x_i$ 
  are those obtained by limits of shifts of various $x_i$. It is not hard to check that the only such sequences are
  in the orbit of either $y_1 = \presuper{\omega}{0}. 1 0^{\omega}$, $y_2 = \presuper{\omega}{0}. 2 0^{\omega}$, or
  $y_3 = \presuper{\omega}{0}. 0^{\omega}$. These all trivially have Turing degree \turdegzero, and so $\spectrum{X} = S\cup\{\turdegzero\}$.
  

  We must now show that $X$ has strong linear complexity by bounding $c_n(X)$ from above.
  Subwords of $y_1$, $y_2$, and $y_3$ have at most one
  symbol not equal to $0$, and there are $2n+1$ $n$-letter words of this type. 
  We now need to find the contribution to $c_n(X)$ from subwords of some $x_i$.
  Fix $i$, and consider any $w$ a subword of $x_i$. 
  \begin{itemize}
  \item if $w$ contains fewer than two symbols from $\{1,2\}$, then
    $w$ has already been counted as a subword of $y_1$, $y_2$, or $y_3$.
  \item if $w$ contains at least two $\{1,2\}$ symbols, then $w$ must occur within \newline
    $0^{n - 1} s^i_1 0^{m_1} s^i_2 0^{m_2} \ldots s^i_k 0^{m_k}$, where 
    $k$ is the unique integer for which $m_{k-1} < n \leq m_{k}$. 
    The number of such $n$-letter subwords of $x_i$ is $\displaystyle \sum_{i=1}^k (1 + m_i)$. 
    If we choose, for instance, $m_i = 2^i - 1$, then this is always strictly less than $4n$. 
  \end{itemize}
  
  So, by combining these counts over all $1 \leq i \leq t$, we see that $c_n(X) < (2n + 1) + 4tn$,
  and therefore $X$ clearly has strong linear complexity.
  If $S$ is a \pizu class, then the set of forbidden words can be
  recursively enumerated, and so $X$ is effectively closed.
  
  
\end{proof}

\thmRealStrongLinear
\begin{proof}
Simply take the union of the subshifts guaranteed by Theorems~\ref{finitecones} and \ref{finiteset}.
\end{proof}

\subsection{Examples with arbitrarily slow superlinear complexity}

\begin{theorem}\label{countablecones}
  Let $S = \{\turdeg{d_0}, \turdeg{d_1}, \ldots\}$ be a countable set of Turing degrees. Then the set of 
  subshifts $X$ with $\spectrum{X} = \bigcup_{k \geq 0} \cone{\turdeg{d_k}}$ exhibits arbitrarily slow superlinear complexity. 
  

\end{theorem}
\begin{proof} 
Choose any increasing unbounded $f: \mathbb{N} \rightarrow \mathbb{N}$. Define $(m_i)$ an strictly increasing (not necessarily computable) sequence of natural numbers satisfying $m_{f(n)/2} > n$ for all sufficiently large $n$; this is possible since $f$ is unbounded.

  Choose a sequence $(\alpha_k)_{k \geq 0}$ where $\alpha_k$ has degree $\turdeg{d_k}$ for $k \geq 0$,
	and where	$\alpha_k$ approaches the limit $\alpha_0$ quickly enough to guarantee $\LL_{m_k}(S_{\alpha_k}) = \LL_{m_k}(S_{\alpha_0})$ for 
	$k \in \mathbb{N}$ (recall that $S_\alpha$ is the Sturmian shift with rotation number $\alpha$); this is possible since the set of reals with any fixed Turing degree is obviously dense.
  Define $X = \bigcup_{k = 0}^{\infty} S_{\alpha_k}$; $X$ is closed (and therefore a subshift) since $\alpha_k \rightarrow \alpha_0$. 
  Since $\spectrum{S_{\alpha_k}} = \cone{\turdeg{d_k}}$ for all $n$, $\spectrum{X} = \bigcup_{k \geq 0}\cone{\turdeg{d}_k}$, as desired.

  It remains to bound $c_n(X)$ from above. Choose any $n$, and define $k$ so that $m_{k-1} \leq n < m_k$. 
  Since $m_{f(n)/2} > n$ for sufficiently large $n$, $k \leq f(n)/2$ for all such $n$. 
  For any $i \geq k$, note that since $\LL_{m_i}(S_{\alpha_i}) = \LL_{m_i}(S_{\alpha_0})$, $\LL_n(S_{\alpha_i}) = \LL_n(S_{\alpha_0})$. 
  Therefore, $c_n(X) \leq k(n+1)$, which is bounded from above by $(n+1)(f(n)/2) \leq nf(n)$ for 
  sufficiently large $n$, completing the proof.

\end{proof}

\begin{theorem}\label{countableset}
  Let $S$ be a countable set of Turing degrees. Then the set of 
  subshifts $X$ with $\spectrum{X} = S \cup \{\turdegzero\}$ exhibits arbitrarily slow
  superlinear complexity.    
\end{theorem}

\begin{proof}
  Let $S = \{\turdeg{d_1}, \turdeg{d_2}, \ldots\}$ be a countable set of Turing degrees, and choose
  $\{0,1\}$-sequences $(s_i)_{i \in \mathbb{N}}$ with $\deg_T(s_i) = \turdeg{d_i}$. Fix any increasing unbounded
  $f: \NN \rightarrow \NN$. Define $(m_i)$ an strictly increasing (not necessarily computable) sequence of natural numbers satisfying
  $m_{(f(n) - 6)/4} > n$ for all sufficiently large $n$; this is possible since $f$ is unbounded.

  Now, for each $i$, define a sequence
  \[
    x_i := \quad \presuper{\omega}{0}. 1 0^{2^1 - 1} 1 0^{2^2 - 1} 1 0^{2^3 - 1} \ldots 1 0^{2^{m_i} - 1} 1 0^{2^{m_i + 1 + s_i(1)}} 
    1 0^{2^{m_i + 2 + s_i(2)}} 1 0^{2^{m_i + 3 + s_i(3)}} \ldots
  \]

  The runs of $0$s in $x_i$ are just consecutive powers of $2$ until the gap of
  length $2^{m_i}$, after which the runs of $0$s switch to increasing powers of
  $2$ which encode the letters of $s_i$. It should be clear from this that
  $\deg_T(x_i) = \deg_T(s_i) = \turdeg{d_i}$. Define a subshift $X$ as the
  closure of the union of the orbits of all $x_i$. All $x \in X$ are then either
  some shift of an $x_i$, or a limit of shifts of various $x_i$. It is easy to
  see that the only such points, apart from the $x_i$ themselves, are in the orbit of either
  
  $y_1 = \presuper{\omega}{0}. 1 0^{2^1 - 1} 1 0^{2^2 - 1} 1 0^{2^3 - 1} \ldots$, $y_2 = \presuper{\omega}{0}. 1 0^{\omega}$, or 
  $y_3 = \presuper{\omega}{0} . 0^{\omega}$. Clearly all of these have Turing degree $\turdegzero$, and so
  $\spectrum{X} = S \cup \{\turdegzero\}$. 

  It remains to bound $c_n(X)$ from above. We first note that for every $n$, the number of
  combined $n$-letter subwords of $y_1$, $y_2$, and $y_3$ is less than or equal to $6n + 1$
  by exactly the same argument as used in Theorem~\ref{finiteset}. Now, choose any $n$, and 
  take the unique $k \in \NN$ for which $m_{k-1} < n \leq m_k$. Recall that for sufficiently large $n$,
  $m_{(f(n) - 6)/4} > n$, and so for such $n$, $k \leq (f(n) - 6)/4$. Then, note that for $i \geq k$, 
  every $n$-letter subword of $x_i$ appears also as a subword of $y_1$. Therefore, we need only 
  count the $n$-letter subwords within $x_i$ for $i < k$. Each $x_i$ is a sequence whose gaps of
  $0$ have lengths which are increasing powers of $2$, and so by a similar argument as in Theorem~\ref{finiteset},
  each $x_i$ has less than or equal to $4k$ subwords of length $n$. Combining the above yields
  \[
    c_n(X) \leq (4k + 6)n,
  \]
  which is less than or equal to $nf(n)$ for large enough $n$, completing the proof.

\end{proof}

\thmRealArbSlow
\begin{proof}
Just take the unions of subshifts guaranteed by Theorems~\ref{finitecones} and \ref{countableset}. 
\end{proof}

\subsection{Examples with arbitrarily computably slow superlinear complexity}


\thmRealArbSlowComp

The proof is composed of the following two lemmas.

\begin{lemma}
 Let $S$ be an arbitrary nonempty closed set of Turing degrees, then the set of subshifts $X$
 with $\spectrum{X}=S\cup\{\turdegzero\}$ exhibits arbitrarily computably slow
 superlinear complexity. 
\end{lemma}
\begin{proof}
  Let $S$ be an arbitrary closed nonempty set of Turing degrees, and choose a
  closed set of $\{1,2\}$-sequences $(s_{\turdeg{d}})_{\turdeg{d} \in S}$ with $\deg_T(s_{\turdeg{d}}) = \turdeg{d}$. 
  Fix any increasing unbounded computable $f: \NN \rightarrow \NN$, and choose a strictly 
  increasing computable function $(m_i)_{i\in\NN}$ so that $m_{(\log_2 f(n))/6} > n$ for all $n$;
  this is possible since $f$ is unbounded and computable.
  We define a subshift $X\subset\{0,1,2\}^\ZZ$ consisting of the closure of the union of 
  the orbits of the following sequences $(x_{\turdeg{d}})_{\turdeg{d} \in S}$:
  \[
    \begin{array}{c}
      x_{\turdeg{d}} := \quad \presuper{\omega}{0} . s_{\turdeg{d}}(1) 0^{m_1} s_{\turdeg{d}}(2) 0^{m_2}s_{\turdeg{d}}(3) 0^{m_3}
      s_{\turdeg{d}}(4) 0^{m_4}\cdots 
    \end{array}
  \] 
  Since $(m_i)$ is computable, it should be clear that $\deg_T(x_{\turdeg{d}}) = \deg_T(s_{\turdeg{d}}) = \turdeg{d}$ for all $\turdeg{d} \in S$. The only    		
  $x \in X$ other than the orbits of the $x_{\turdeg{d}}$ 
  are those obtained by limits of shifts of various $x_{\turdeg{d}}$. 
  Since $\left(s_{\turdeg d}\right)_{\turdeg d\in S}$ is closed, the only such points outside of the
  $x_{\turdeg{d}}$ themselves are
  in the orbit of either $y_1 = \presuper{\omega}{0}. 1 0^{\omega}$, $y_2 = \presuper{\omega}{0}. 2 0^{\omega}$, or
  $y_3 = \presuper{\omega}{0}. 0^{\omega}$. These all trivially have Turing degree \turdegzero, and so $\spectrum{X} = S\cup\{\turdegzero\}$.

  We must now bound $c_n(X)$ from above. 
  We recall from the proof of Theorem~\ref{finiteset} 
  that any particular $x_{\turdeg{d}}$ has complexity less than or equal to $6n$; let's for instance fix $x$ to be the sequence 
  induced by $s_{\turdeg{d}} = 1^{\omega}$. Now, consider any $n$, and define $k$ so that $m_{k-1} < n \leq m_k$. 
  Recall that for sufficiently large $n$, $m_{(\log_2 f(n))/6} > n$, and so for such $n$, $k \leq (\log_2 f(n))/6$. 
  Clearly any $n$-letter word in $\LL(X)$ has 
  less than or equal to $k$ symbols in $\{1,2\}$, and is obtained by taking some subword of $x$ and changing the $1$s to
  a combination of $1$s and $2$s, which can be done in less than or equal to $2^k$ ways. 
  Therefore, $c_n(X) \leq 2^{k} (6n)$, which is less than $nf(n)$ for sufficiently large $n$,
  completing the proof.
  
\end{proof}
\begin{lemma}\label{lem:compslowsublin}
 Let $S$ be an arbitrary closed nonempty set of Turing degrees, then the set of subshifts $X$
 with $\spectrum{X}=\bigcup_{\turdeg{d}\in S}\cone{\turdeg d}$ exhibits arbitrarily computably slow
 superlinear complexity. 
\end{lemma}
\begin{proof} This proof is based on a construction of \citet{miller}.
  Let $S$ be an arbitrary non-empty closed set of Turing degrees and $f:\NN\to\NN$ be any
  increasing computable unbounded function. Choose $g:\NN\to\NN$ a computable
  function satisfying $g(n) > 2$ and $2^{2n + 5} \leq f(g(n))$ for all $n$. Take
  $\left(s_{\turdeg d}\right)_{\turdeg d\in S}$ to be a closed set of $\{0,1\}$ sequences
  with $s_{\turdeg{d}}$ of degree $\turdeg{d}$.
  

  Now, for each finite word $\sigma\in\{0,1\}^*$, we inductively define two words $a_\sigma$ and
  $b_\sigma$ as follows: (here $\epsilon$ represents the empty word)
  
  \begin{minipage}{.5\linewidth}
    \begin{align*}
      a_\epsilon ={} & 1 \\
      a_{\sigma 0} ={} & b_{\sigma}\underbrace{a_{\sigma}\cdots a_{\sigma}}_{g(|\sigma|)-1} \\
      a_{\sigma 1} ={} & a_{\sigma}\underbrace{b_{\sigma}\cdots b_{\sigma}}_{g(|\sigma|)-1}
    \end{align*} 
  \end{minipage}
  \begin{minipage}{.5\linewidth}
    \begin{align*}
      b_\epsilon ={} & 0 \\
      b_{\sigma 0} ={} & a_{\sigma}\underbrace{b_{\sigma}\cdots b_{\sigma}}_{g(|\sigma|)} \\
      b_{\sigma 1} ={} & b_{\sigma}\underbrace{a_{\sigma}\cdots a_{\sigma}}_{g(|\sigma|)}
    \end{align*} 
  \end{minipage}
  Note that $a_\sigma$ is always a prefix of $b_\sigma$, and so for every $\sigma$, all subwords of
  $a_\sigma$ appear in $b_\sigma$.
  Now let $\mathcal{F}_S$ be the set of words that are not subwords of any word of the form $b_\sigma$
  for $\sigma$ a prefix of some $s_{\turdeg{d}}$. Now let $X = X(\mathcal{F}_S) \subset\{0,1\}^\NN$ be the
  subshift defined by the set $\mathcal{F}_S$ of forbidden words. 

  We first show that $\spectrum X$ is $\bigcup_{\turdeg d\in S}\cone{\turdeg
    d}$. By definition of $\mathcal{F}_S$, for every $x\in X$ and $n\in\NN$, 
    $x$ is a concatenation of $a_{\sigma}$ and $b_{\sigma}$ for some $\sigma$
    of length $n$, and $\sigma$ is the prefix of some $s_{\turdeg{d}}$. By the recursive
    definition of the words $a_{\sigma}$ and $b_{\sigma}$ and the fact that 
    $\left(s_{\turdeg d}\right)_{\turdeg d\in S}$ is closed, any fixed $x$ is induced
    in this way by prefixes of a single $s_{\turdeg{d}}$. Furthermore, this $s_{\turdeg{d}}$ 
    can be computed from $x$ by just counting the number of successive occurences of 
    $a_\sigma$ and $b_\sigma$. Start with $\sigma=\epsilon$. Then, if the sequence 
    contains two consecutive $a_\epsilon$, then the first bit of $s_{\turdeg{d}}$ is 
    $0$, and otherwise it is $1$. From this information, we can
    deduce subsequent bits of $s_{\turdeg{d}}$ in a similar fashion. 
    Therefore, each sequence of $X$ has degree above some degree of $S$, and so
    $\spectrum X \subseteq \bigcup_{\turdeg d\in S}\cone{\turdeg{d}}$.
    
    For the opposite containment, we fix any $s_{\turdeg{d}}$, and define $\sigma_n$
    to be its prefix of length $n$ for every $n \in \NN$. The sequence $b_{\sigma_n}$ 
    approaches a limit $x_{\turdeg{d}}$, which is in $X$ by definition. Clearly
    $\deg_T(x_{\turdeg{d}}) = \turdeg{d}$, and $x_{\turdeg{d}}$ is recurrent.
    Therefore, by~\Cref{lem:recurrentcone}, $\spectrum X$ contains $\cone{\turdeg{d}}$.
    Since $\turdeg{d} \in S$ was arbitrary, the reverse containment is proved, and so
    $\spectrum X=\bigcup_{\turdeg d\in S}\cone{\turdeg d}$. 

It remains to bound $c_n(X)$ from above. We first note that for every $\sigma$
of length $k$, every word $a_{\sigma}$, $b_{\sigma}$ has length between $\prod_{i=0}^{k-1} g(i)$ and 
$\prod_{i=0}^{k-1} (g(i) + 1) \leq 2^k \prod_{i=0}^{k-1} g(i)$. For convenience, 
we denote $h(k) = \prod_{i=0}^{k-1} g(i)$ for all $k$. Now, for any $n$, there 
exist $k$ and $1 \leq j < g(k)$ so that $jh(k) \leq n < (j+1) h(k)$. 
For every point of $x$, there exists $\sigma$ of length $k$ so that $x$ 
is a concatenation of $a_{\sigma}$ and $b_{\sigma}$, and there are either
at least $g(k) - 1$ copies of $a_{\sigma}$ between any two $b_{\sigma}$ or vice versa.
Therefore, every word in $\LL_{n}(X)$ is either a subword of 
$(a_{\sigma})^{j+1} b_{\sigma} (a_{\sigma})^{j}$ or 
$(b_{\sigma})^{j+1} a_{\sigma} (b_{\sigma})^{j}$ (recall that $j < g(k)$,
$n < (j+1) h(k)$, and all $a_{\sigma}$, $b_{\sigma}$ have length at least $h(k)$.)
For fixed $\sigma$, the number of such words is clearly bounded from above by 
the sum of the lengths of these words, which is less than or equal to $(4j+4) 2^k h(k)$. 
Since there are $2^k$ possible $\sigma$, we see that $c_n(X) \leq (4j+4) 2^{2k} h(k)
\leq 2^{2k + 3} n \leq f(g(k-1)) n$ and $f(g(k-1))n\leq n f(n)$ since $g(k-1) \leq h(k)\leq n$, completing the proof.

\end{proof}
\subsection{Examples with arbitrarily computably fast subexponential complexity}

\thmRealArbFastComp
\begin{proof}
  Let $S$ be an arbitrary nonempty closed set of Turing degrees, and choose a
  closed set of $\{2,3\}$-sequences $(s_{\turdeg{d}})_{\turdeg{d} \in S}$ with
  $\deg_T(s_{\turdeg{d}}) = \turdeg{d}$.
  Fix any increasing unbounded computable $f: \NN \rightarrow \NN$, and choose a strictly increasing computable function $(m_i)_{i\in\NN}$ so that 
  $m_{f(n)} > n$ for sufficiently large $n$; this is possible since $f$ is unbounded and computable. For each $i$, 
  define $S_i \subset \{0,1\}^{m_i}$ to be the set of all $\{0,1\}$ words of length $m_i$ which 
  do not contain $1$s separated by a distance of less than $i$. (For instance, if $m_2 = 3$, then 
  $S_2 = \{000, 001, 010, 100, 101\}$.) For each $i$, define $w_i$ to be the concatenation of
  all words in $S_i$ in lexicographic order, followed by a second copy of $0^{m_i}$. 
  (In the above example, $w_2$ would be $000001010100101000$.) Since $(m_i)$ is computable,
  the sequence $(w_i)$ is also computable.

  We define a subshift $X\subset\{0,1,2,3\}^\ZZ$ consisting of the closure of the union of 
  the orbits of the following sequences $(x_{\turdeg{d}})_{\turdeg{d} \in S}$:
  \[
    \begin{array}{c}
      x_{\turdeg{d}} := \quad \presuper{\omega}{0} . s_{\turdeg{d}}(1) w_1 s_{\turdeg{d}}(2) w_2 s_{\turdeg{d}}(3) w_3 s_{\turdeg{d}}(4) w_4 \cdots 
    \end{array}
  \] 
  Since $(w_i)$ is computable, $\deg_T(x_{\turdeg{d}}) = \deg_T(s_{\turdeg{d}})
  = \turdeg{d}$ for all $\turdeg{d} \in S$. The only $x \in X$ other than the orbits of the $x_{\turdeg{d}}$ 
  are those obtained by limits of shifts of various $x_{\turdeg{d}}$. The only
  such points are either shifts of some $x_{\turdeg{d}}$ by closedness of $\left(s_{\turdeg d}\right)_{\turdeg d\in S}$ or
  in the orbit of either $y_1 = \presuper{\omega}{0}. 1 0^{\omega}$, $y_2 = \presuper{\omega}{0}. 2 0^{\omega}$, or
  $y_3 = \presuper{\omega}{0}.3 0^{\omega}$. These all trivially have Turing degree $\turdegzero$, and so $\spectrum{X} = S\cup\{\turdegzero\}$.

If $h(X)$ were positive, then by \Cref{periodic}, there would exist an ergodic $\mu$ with the property that $\mu$-a.e. $x \in X$ is aperiodic. 
However, the letters $1$, $2$, and $3$ have frequency zero in every point of $X$, and so applying \Cref{birkhoff} with $w = 0$ shows that
$\mu([0]) = 1$. By $\sigma$-invariance of $\mu$, we see that $\mu(\{0^{\omega}\}) = 1$, i.e. $\mu$-a.e. $x \in X$ is periodic, a contradiction. 
Therefore, our original assumption was false and $h(X) = 0$.

  Finally, we must bound $c_n(X)$ from below. For any $n$, choose $k$ such that $n_{k-1} < n \leq n_k$.
  Recall that for sufficiently large $n$, 
  $m_{f(n)} > n$, and so for such $n$, $k \leq f(n)$.
  Every word of length $n$ which does not contain $1$s separated by distance less than $k$ is a subword of
  $w_k$, and so is in $\LL(X)$. There are at least $2^{\lceil n/k \rceil}$ such words; 
  for instance, any word which contains $0$ at all indexes except multiples of $k$ has this property.
  Therefore, 
  \[
    c_n(X) \geq 2^{\lceil n/k \rceil} \Longrightarrow \frac{\log c_n(X)}{n} \geq \frac{1}{k},
  \]
  which is greater than $\frac{1}{f(n)}$ for all sufficiently large $n$, completing the proof.

\end{proof}

\subsection{Examples with arbitrarily fast subexponential complexity}

\thmRealArbFast

\begin{proof}
Consider any such set $S$, a subshift $Y$ (with alphabet $\A$) as in \Cref{lem:compslowsublin} with $\spectrum{Y} = \bigcup_{\turdeg{d}\in S}\cone{\turdeg d}$ and subexponential complexity, and any increasing unbounded $f: \NN \rightarrow \NN$. Choose a sequence $(n_k)$ so that $n_{k+1} < kf(k)$ for every $k$, and for which the sequence $(n_{k+1} - n_k)$ approaches infinity; this is possible since $f(k)$ is increasing and unbounded (for instance, we could define $n_k = \lfloor kf(k) \rfloor$).

Now, define a subshift $Z$ with alphabet $\{0,1\}$ as follows. Define the set $C$ of $\{0,1\}$-sequences which are $0$ at all locations not indexed by any $n_k$ 
(and which can be $0$ or $1$ at all $(n_k)$-indexed locations). Then, take $Z = \overline{\bigcup_{c \in C} \mathcal{O}(c)}$, and note that trivially the sequence $0^{\omega}$ is in $Z$. The same argument as was used in the proof of \Cref{thm:fastcompsubexp} shows that $h(Z) = 0$, i.e. that $Z$ has subexponential complexity.

Now, just define $X = Y \times Z$, the subshift on $\A \times \{0,1\}$ of sequences projecting to a point of $Y$ along the first coordinate and to a point of $Z$ along the second. For every $y \in Y$, the point $y \times \{0^{\omega}\}$ is in $X$, and clearly has the same degree as $y$. Therefore, 
$\bigcup_{\turdeg{d}\in S}\cone{\turdeg d} = \spectrum{Y} \subseteq \spectrum{X}$. For the reverse containment, note that all points of $X$ project to points of 
$Y$ along the first coordinate, and so for every $x \in X$, there exists $y \in Y$ s.t. $\deg_T x \geq_T \deg_T y$. Since $\deg_T y \in \spectrum{Y} = \bigcup_{\turdeg{d}\in S}\cone{\turdeg d}$, by definition of cone, $\deg_T x \in \bigcup_{\turdeg{d}\in S}\cone{\turdeg d}$ as well. Since $x \in X$ was arbitrary, the reverse containment is shown, and so $\spectrum{X} = \bigcup_{\turdeg{d}\in S}\cone{\turdeg d}$. 

It remains only to bound the complexity of $X$. By definition, $\LL_n(X) =
\LL_n(Y) \times \LL_n(Z)$, and since both $Y$ and $Z$ have subexponential
complexity, this means that $X$ does as well. For a lower bound, choose any $n
\geq n_1$, and choose $k$ s.t. $n_k \leq n < n_{k+1}$. Then every $w \in
\{0,1\}^n$ which has $0$ at all locations not equal to some $n_1, \ldots, n_k$
is in $\LL_n(Z)$, and so $c_n(Z) \geq 2^k$. Clearly then $c_n(X) \geq 2^k$ as well, and so
\[
\frac{\log c_n(X)}{n} \geq \frac{\log 2^k}{n} \geq \frac{k}{n_{k+1}} > \frac{1}{f(n_k)} \geq \frac{1}{f(n)}.
\]
Since $f$ was arbitrary, this completes the proof.

\end{proof}

\subsection{Examples with intermediate complexity}

We begin by showing that subshifts can get close to any ``intermediate'' computable complexity function with any Turing spectrum containing $\turdegzero$. 

\thmRealIntermediate
We need the following very simple lemma. 

\begin{lemma}\label{lem:wordcount}
For every alphabet $\A$, $n \in \mathbb{N}$, and $w \in \A^*$, denote by $N$ the number of $n$-letter subwords of $w$. Then, for every $1 \leq k \leq N$, there exists a prefix $p$ of $w$ containing exactly $k$ different subwords of length $n$. 
\end{lemma}

\begin{proof}
If we define $f: [n, |w|] \rightarrow \mathbb{N}$ by taking $f(i)$ to be the number of different $n$-letter subwords of $w(1) \ldots w(i)$, the following facts are clear:
$f(n) = 1$, $f(|w|) = N$, and $f(i + 1) - f(i) \in \{0,1\}$ for all $i$. It follows immediately that for every $1 \leq k \leq N$, there exists $m$ so that 
$f(m) = k$; taking $p = w(1) \ldots w(m)$ then satisfies the conditions of the lemma.
\end{proof}

\begin{proof}[Proof of \Cref{thm:intermediate}]
	
	Let $S$ be an arbitrary nonempty closed set of Turing degrees, and choose a
  closed set $(s_{\turdeg{d}})_{\turdeg{d} \in S}$ of $\{4,5\}$-sequences with $\deg_T(s_{\turdeg{d}}) = \turdeg{d}$. 
	Then, inductively defining a computable sequence $(n_k)_{k \in \NN}$ so that
	$\frac{\log g(n_{k+1})}{n_{k+1}} < \frac{1}{n_k}$ and $n_k | n_{k+1}$ for all $k$ (here we use the fact that $g$ is computable and subexponential). 
	By passing to a subsequence if necessary, we may also assume that $n_k > 2^k \sum_{i=1}^{k-1} n_i (2 + 2^i)$ for all $k$.
	Since $g$ is subexponential, we may also assume that $g(n_k) < 2^{n_k}$ for all $k$.
	Define $w_1$ by concatenating all words in $\{1,2\}^{n_1}$ in lexicographic order; then, since $g(n_1) < 2^{n_1}$, we can use \Cref{lem:wordcount} to construct 
	$p_1 \in \{1,2\}^*$ containing exactly $g(n_1)$ words of length $n_1$.
	For $k \geq 1$, define $w_{k+1}$ by concatenating all words (again in lexicographic order) of the form $0^{n_k - 1} \ell_1 0^{n_k - 1} \ell_2 \ldots 0^{n_k - 1} \ell_{n_{k+1}/n_k}$ for 
	$\ell_i \in \{1,2\}$; clearly then $w_{k+1}$ contains at least $2^{n_{k+1}/n_k}$ words of length $n_{k+1}$. By definition, $g(n_{k+1}) < 2^{n_{k+1}/n_k}$, and so we 
	can use \Cref{lem:wordcount} to construct a prefix $p_{k+1}$ of $w_{k+1}$ which contains exactly $g(n_{k+1})$ words of length $n_{k+1}$. Note that since $g$ and $(n_k)$ were chosen computable, and since the procedure from the proof of \Cref{lem:wordcount} is clearly algorithmic, the sequence $(p_k)$ is computable as well. 
	For future reference, we note that $|w_k| = n_k 2^{n_k/n_{k-1}} <  n_k 2^{n_k}$ for all $k$, and so $|p_k| < n_k 2^{n_k}$ as well.
	Then, for $\turdeg{d} \in S$, define
	\[
	\begin{array}{c}
      x_{\turdeg{d}} := \quad \presuper{\omega}{0} . p_1 (0^{n_1} s_{\turdeg{d}}(1) 0^{n_1-1}) p_2 (0^{n_2} s_{\turdeg{d}}(2) 0^{n_2-1}) p_3 (0^{n_3} 				   	s_{\turdeg{d}}(3) 0^{n_3-1}) p_4 \cdots 
  \end{array}
  \]
	
	We define the subshift $X\subset\{0,1,2,3,4\}^\ZZ$ to be the closure of the union of 
  the orbits of $(x_{\turdeg{d}})_{\turdeg{d} \in S}$. Since the sequence
  $(p_k)$ is computable, $\deg_T(x_{\turdeg{d}}) = \deg_T(s_{\turdeg{d}}) =
  \turdeg{d}$ for all $\turdeg{d} \in S$. The only $x \in X$ other than the  
	orbits of the $x_{\turdeg{d}}$ 
  are those obtained by limits of shifts of various $x_{\turdeg{d}}$. Since $\left(s_{\turdeg d}\right)_{\turdeg d\in S}$ is closed, the only such points begin with $\presuper{\omega}{0}$ and end with $0^{\omega}$, and so 
	all trivially have Turing degree \turdegzero. Therefore, $\spectrum{X} = S\cup\{\turdegzero\}$. It remains only to prove the desired estimates on $c_{n_k}(X)$.
	
	Clearly, $c_{n_k}(X) \geq g(n_k)$ for every $k$, since $p_k$ contains $g(n_k)$ subwords of length $n_k$. For an upper bound, we consider an arbitrary 
	$n_k$-letter subword $w$ of some $x^{\turdeg{d}}$, and break into several cases. If the final letter of $w$ is to the right of the $0^{n_k}$ immediately 
	following $p_k$, then $w$ has at most one non-$0$ symbol; there are at most $1 + 4n_k$ options for such words. If the final letter of $w$ is inside the $0^{n_k}$ 
	immediately
	following $p_k$, then the location of that final letter determines $w$, and so there are at most $n_k$ options for $w$ in that case. If the initial letter of $w$  is to the left of $p_k$, then as long as $w$ is not the word $0^{n_k}$ (which was already counted in a previous case), it is determined by the location of its rightmost letter and the values of
	$s^{\turdeg{d}_i}$ for $1 \leq i \leq k$, giving not more than $2^k \sum_{i=1}^{k-1} |p_i| + 2n_i < 2^k \sum_{i=1}^{k-1} n_i(2 + 2^{n_i})$ possibilities. Combining all of this yields
	\[
	c_{n_k}(X) \leq g(n_k) + 5n_k + 1 + 2^k \sum_{i=1}^{k-1} n_i (2 + 2^{n_i}) \leq g(n_k) + 6n_k.
	\]
	So, $g(n_k) \leq c_{n_k}(X) < g(n_k) + 6n_k$, which implies that $\lim_{k \rightarrow \infty} \frac{c_{n_k}(X)}{g(n_k)} = 1$ since $g$ was assumed to be
	superlinear.

\end{proof}

We now prove that every union of cones can be realized as a Turing spectrum for a slightly smaller range of intermediate complexities, and begin by proving that cones can be realized by such subshifts.

\begin{lemma}\label{lem:intermediate}
Let $\turdeg{d}$ be any Turing degree, and let $g: \mathbb{N}
  \rightarrow \mathbb{N}$ be any computable function
	which is superquadratic (i.e. $\frac{g(n)}{n^2} \rightarrow \infty$) and subexponential (i.e. $\frac{\log g(n)}{n} \rightarrow 0$). 
Then there exists a subshift $X$ with $\spectrum{X} = \cone{\turdeg d}$ and a sequence $(n_k)$ so that 
\[
\lim_{k \rightarrow \infty} \frac{c_{n_k}(X)}{g(n_k)} = 1.
	\]
\end{lemma}

\begin{proof}
Consider any such $\turdeg{d}$ and $g$. Choose any $\alpha$ for which $\deg_T \alpha = \turdeg{d}$, and define the Sturmian subshift $S_\alpha$; recall that $\spectrum{S_{\alpha}} = \cone{\turdeg{d}}$. 

Since $g$ is superquadratic, computable, and subexponential, the function $h(n) := \lfloor g(n)/n \rfloor$ is superlinear, computable, and subexponential. Therefore, by \Cref{thm:intermediate}, there exists a subshift $Y$ with $\spectrum{Y} = \{\turdegzero\}$ and a sequence $n_k$ for which $\lim_{k \rightarrow \infty} \frac{c_{n_k}(Y)}{h(n_k)} = 1$. Now, simply define $X = Y \times S_{\alpha}$. Clearly $\LL_{n_k}(X) = \LL_{n_k}(Y) \times \LL_{n_k}(S_{\alpha})$, and so
$c_{n_k}(X) = c_{n_k}(Y) c_{n_k}(S_{\alpha}) = (n_k + 1) c_{n_k}(Y)$. However, this clearly implies that
\[
\lim_{k \rightarrow \infty} \frac{c_{n_k}(X)}{g(n_k)} = 
\lim_{k \rightarrow \infty} \frac{c_{n_k}(X)/n_k}{g(n_k)/n_k} = 
\lim_{k \rightarrow \infty} \frac{c_{n_k}(Y)}{h(n_k)} = 1.
\]
Finally, we note that every point of $X$ projects to a point of $S_{\alpha}$ along the second coordinate, and so since $\spectrum{S_{\alpha}} = \cone{\turdeg{d}}$, $\spectrum{X} \subseteq \cone{\turdeg{d}}$. For the reverse containment, note that $Y$ contains a computable point $y$, and so for every $s \in S_{\alpha}$, $x = y \times s \in X$, and $\deg_T x = \deg_T s$. Therefore, $\spectrum{X} = \cone{\turdeg{d}}$, completing the proof.

\end{proof}

\thmRealIntermediat


\begin{proof}
Consider any such $S$, $g$, and an arbitrary $\turdeg{d} \in S$. By \Cref{lem:intermediate}, there exists a subshift $Y$ with 
$\lim_{k \rightarrow \infty} \frac{c_{n_k}(Y)}{g(n_k)} = 1$ and $\spectrum{Y} = \cone{\turdeg{d}}$. Since 
$g$ is computable and superlinear, by \Cref{lem:compslowsublin} there exists a subshift $Z$ with 
$\lim_{n \rightarrow \infty} \frac{c_{n}(Z)}{g(n)} = 0$ and $\spectrum{Z} = \bigcup_{\turdeg{d}\in S}\cone{\turdeg d}$. 
Then, we simply define $X = Y \cup Z$. It's immediate that $\spectrum{X} = \spectrum{Y} \cup \spectrum{Z} = \bigcup_{\turdeg{d}\in S}\cone{\turdeg d}$,
and
\[
\lim_{k \rightarrow \infty} \frac{c_{n_k}(X)}{g(n_k)} = 
\lim_{k \rightarrow \infty} \frac{c_{n_k}(Y)}{g(n_k)} +
\lim_{k \rightarrow \infty} \frac{c_{n_k}(Z)}{g(n_k)} = 1.
\]

\end{proof}

\subsection{Examples with exponential complexity/positive entropy}

Though it is still not known exactly which sets of Turing degrees can be the spectrum of a subshift, we can show that there are no restrictions beyond those of Theorem~\ref{posent} on the spectra of positive entropy subshifts.

\thmRealPosEntropy
\begin{proof}
  Consider $X$ as in the theorem, and $h'>h$ with $\deg_T h' = \turdeg d$ and
  suppose that $\cone{\turdeg{d}} \subset   \spectrum{X}$. There exists an
  integer $k$ such that $h'/k\in (0,1)$, choose
  $\alpha=h'/k$, clearly $\deg_T(\alpha) = \turdeg{d}$. Define the Sturmian
  subshift $S_{\alpha}$ with rotation number $\alpha$, and recall
  that $\spectrum{S_{\alpha}} = \cone{\deg_T \alpha} = \cone{\turdeg{d}}$.
  Define $Y \subset S_{\alpha} \times \{a,b\}^{\mathbb{Z}}$ to be the subshift with
  alphabet $\{(0, a), (1, a), (1, b)\}$ of sequences which project to some 
  point of $S$ along the first coordinate (the lack of $(0, b)$ in the alphabet
	imposes the additional constraint that only $a$ can appear with $0$).

  Let us compute the entropy of $Y$ by bounding $c_n(Y)$. In every word of
  length $n$ appearing in $Y$, the number of $1$s is either $\lfloor n\alpha \rfloor$ or
  $\lfloor n\alpha \rfloor+1$, and for each $1$ there are two choices for the
  second component and for each $0$ only one choice. Therefore,
  \[
    (n+1) 2^{\lfloor n\alpha \rfloor}\leq c_n(Y)\leq (n+1) 2^{\lfloor n\alpha \rfloor+1}\text{, and}
	\]
  \[ 
    \frac{\log\left( (n+1) 2^{\lfloor n\alpha \rfloor}\right)}{n}\leq \frac
		{\log c_n(Y)}{n}\leq \frac{\log\left((n+1) 2^{\lfloor n\alpha\rfloor+1}\right)}{n}\text{,}
  \]
  and by taking logs, dividing by $n$, and letting $n \rightarrow \infty$, we see that $h(Y)=\alpha$.

  We now claim that $\spectrum{Y} = \cone{\turdeg{d}}$. Clearly, for all $y \in Y$,
  one can compute a point $s \in S_{\alpha}$ by projecting to the first
  coordinate,  and so $\deg_T y \geq_T \deg_T s \geq_T \turdeg{d}$. Therefore,
  $\spectrum{Y} \subseteq \cone{\turdeg{d}}$.
  For the opposite containment, take
  any  $\turdeg{d'} \geq \turdeg{d}$. We can choose $z \in \{a,b\}^{\ZZ}$ with $\deg_T(z) = \turdeg{d'}$ and
  $s \in S_{\alpha}$ with $\deg_T(s) = \turdeg{d'}$, 		   since
  $\spectrum{S_{\alpha}} = \cone{\turdeg{d}}$. Define $y \in Y$ as the point
  with $s$ occupying the first coordinate, and $z$ occupying the locations in 
	the second coordinate where $s$ has $1$s. 
  Then $\deg_T(y) = \sup(\deg_T(s), \deg_T(z)) =
  \turdeg{d'}$. Therefore, $\cone{\turdeg{d}} \subseteq \spectrum{Y}$, and so
  indeed $\spectrum{Y} = \cone{\turdeg{d}}$.

  Now, we consider the $k$-fold Cartesian product $Y^k$ of $Y$, 
  i.e. the subshift consisting of ordered $k$-tuples of points of $Y$
	(viewed as a sequence over the alphabet $\{(0, a), (1, a), (1, b)\}^k$.)
	Clearly $h(Y^k)=k\alpha =h'$. Since points of $Y^k$ come from $k$-tuples of points of $Y$,
	it's easy to see that 
	\[
	\spectrum{Y^k} = \{\sup(\deg_T(y_1), \ldots \deg_T(y_k)) \ : \ y_1, \ldots, y_k \in Y\}.
	\]
	Since $\spectrum{Y} = \cone{\turdeg{d}}$, we see that $\spectrum{Y^k}=\cone{\turdeg{d}}$ as well. 
	Finally, define $P = X \sqcup Y^k$. By the preceding, $\spectrum{P} =
  \spectrum{X} \cup \spectrum{Y} = \spectrum{X} \cup \cone{\turdeg{d}} =
  \spectrum{X}$. Finally, since $P$ is a disjoint union of $X$ and $Y^k$, its
  entropy is $h(P)=\max \left( h(X),h(Y^k) \right)=h'$.

\end{proof}

\section{An example of spectrum with an isolated degree}\label{S:weirdSpec}

All our previous realizations of spectra are of either unions of cones or an arbitrary set
of Turing degrees containing \turdegzero. We now construct an example of
subshift whose spectrum is not a union of cones and does not contain \turdegzero.

\thmWeirdSpec
\begin{proof}
  Let $\alpha\in(0,1)$ be of degree $\turdeg{d_1}$,
  by Schoenfield's limit lemma (see for example \citet[Proposition
  10.5.10]{Cooper}), there exists a computable sequence of rationals
  $(r_n)_{n\in\NN}$ such that $r_n\to \alpha$. Let $f(n)$ be a strictly
  increasing function of degree $\turdeg{d_2}$.

  Let us now consider the following
  sequence:
  \[
    \begin{array}{c}
    z:=\quad \cdots 2 w_6 2 w_4 2 w_2 2 w_1 2 w_3 2 w_5 2\cdots
    \end{array}
  \]
  where each $w_i\in\{0,1\}^{f(i)}$ and where $w_i(k)=\left\lfloor
    (k+1)r_i\right\rfloor-\left\lfloor k r_i
  \right\rfloor$ for $k\in\{1,\dots,f(i)\}$. We claim that the degree of $z$ is
  $\turdeg{d_2}$: since $(r_n)_{n\in\NN}$ is computable, using $\turdeg{d_2}$ as an
  oracle, it is easy to compute $z$, while $\turdeg{d_2}$ can be computed from $z$
  just by computing the distances between subsequent 2s.

  Now let $X$ be the orbit closure of $z$. We claim that
  $X=S_{\alpha}\cup\left\{ \sigma^i(z)\mid i\in\ZZ \right\}$.

  It is clear that
  each word $w$ of $\LL\left( S_\alpha \right)$ is in $X$: since $r_n\to\alpha$,
  there exists some $k$ such that $w$ is a subword of $w_k$. So
  $S_\alpha\subseteq X$.

  Conversely take a word $w$ appearing in some $w_k$, $w$ appears in a sequence
  of $X$ which is not a translate of $z$ iff there is an infinite increasing sequence
  $(s_i)_{i\in\NN}$ such that $w$ is a subword of $w_{s_i}$. This is true iff
  $w\in\LL(S_\alpha)$. Therefore $X\subseteq S_\alpha\cup\{\sigma^i(z)\mid i\in\ZZ\}$. 

  Now it is clear that $\spectrum{X}=\cone{\turdeg{d_1}}\cup\{\turdeg{d_2}\}$.
\end{proof}
Note that the proof can be generalized to use the assumption that
$\turdeg{d_1} \leq \turdeg a'$ and $\turdeg{d_2} \geq \turdeg a$ for any degree 
$\turdeg a$; the only change is that the sequence $(r_n)$ will be computable
using $a$, a sequence of degree $\turdeg a$, as an
oracle rather than computable. 

\printbibliography
\end{document}